\documentclass[11pt]{article}
\usepackage{ieee}
\usepackage{authblk}
\usepackage{comment}
\usepackage{bm}
\usepackage{url}
\usepackage{algorithm}
\usepackage{algorithmic}
\usepackage{nth}
\usepackage{amssymb}
\usepackage{amsthm}
\usepackage{epsfig}
\usepackage{graphicx}
\usepackage{wrapfig}
\usepackage{enumitem}
\usepackage{subfig}
\usepackage{color}
\usepackage{multirow}
\usepackage[dvipsnames]{xcolor}
\definecolor{applegreen}{rgb}{0.55, 0.71, 0.0}

\newtheorem{theorem}{Theorem}
\newtheorem{lemma}[theorem]{Lemma}

\usepackage{graphicx}
\usepackage{balance}

\begin{document}

\title{Structure-Preserving Community In A Multilayer Network: Definition, Detection, And Analysis}
\author[1]{Abhishek Santra\thanks{abhishek.santra@mavs.uta.edu}}
\author[1]{Kanthi Sannappa Komar\thanks{kanthisannappa.komar@mavs.uta.edu}}
\author[2]{Sanjukta Bhowmick\thanks{sbhowmick@unomaha.edu}}
\author[1]{Sharma Chakravarthy\thanks{sharma@cse.uta.edu}}
\affil[1]{IT Lab, CSE Department, University of Texas at Arlington, Texas, USA}
\affil[2]{Department of Computer Science, University of North Texas, Denton, Texas, USA}

\maketitle

\begin{abstract}
Multilayer networks or MLNs (also called  multiplexes or network of networks) are being used extensively for modeling and analysis of data sets with multiple entity and feature  types as well as their relationships. As the concept of communities and hubs are used for these analysis, a \textbf{structure-preserving} definition for them on MLNs (that retains the original \textit{MLN structure} and node/edge labels and types) and its efficient detection are critical. There is no structure-preserving definition of a community for a MLN as most of the current analyses aggregate a MLN to a single graph.

Although there is consensus on community definition for single graphs (and detection packages) and to a lesser extent for homogeneous MLNs, it is lacking for heterogeneous MLNs. In this paper, we not only provide a structure-preserving definition for the first time, but also its efficient computation using a decoupling approach, and discuss its characteristics \& significance for analysis.

The proposed decoupling approach for efficiency combines communities from individual layers to form a \textit{serial k - community} for connected k layers in a MLN. We propose several weight metrics for composing layer-wise communities using the bipartite graph match approach based on the analysis semantics. Our proposed approach has a number of advantages. It: i) leverages extant single graph community detection algorithms, ii) is based on the widely-used maximal flow bipartite graph matching for composing k layers, iii) introduces several weight metrics that are customized for the community concept, and iv) experimentally validates the definition, mapping, and efficiency from a flexible analysis perspective on widely-used IMDb data set.
\\

\noindent \textbf{Keywords:} Heterogeneous Multilayer Networks; Bipartite Graphs; Community Definition and Detection; Decoupling-Based Composition
\end{abstract}

\section{Motivation}
\label{sec:introduction}
As data sets become more complex in terms of entity and feature types, and large, the approaches needed for their modeling and analysis also warrant extensions or new alternatives to match the data set complexity and scale. With the advent of social networks and large data sets, we have already seen a surge in the use of graph-based modeling along with a renewed interest in concepts, such as community and hubs  needed for their analysis. Even data sets that may not be inherently graph-based (e.g., IMDb data set used in this paper) may benefit from the use of graph representation for modeling (from an understanding perspective) and for performing various kinds of analysis that may be difficult or not possible using the traditional Database Management System (DBMSs) or mining approaches.

Data sets being analyzed may also contain features that are derived in addition to explicit ones. Derived information can also be used as an attribute for analysis purposes. We have applied content extraction using deep neural network model~\cite{CICLing:Vu2019} on Facebook status updates to infer privacy-concern corresponding to personality traits for an individual to analyze feature correlation with other attributes such as age.

Informally, MLNs\footnote{The terminology used for variants of multilayer networks varies drastically in the literature and many a times is not even consistent with one another. For clarification, please refer to~\cite{MultiLayerSurveyKivelaABGMP13} which provides an excellent comparison of terminology used in the literature, their differences, and usages clearly.} are \textit{layers of networks (also called network of networks)} where each layer is a simple graph and captures the semantics of an attribute or feature of an entity type. The layers can also be connected. If each layer of a MLN has the \texttt{same set of entities of the same type}, it is termed a homogeneous MLN (or HoMLN.) For a HoMLN, intra-layer edges are shown explicitly and inter-layer edges are not shown as they are implicit. If \texttt{the set and types of entities are different for each layer}, then relationships of entities across layers are shown using explicit inter-layer edges (in addition to intra-layer edges.) This distinguishes a heterogeneous MLN (or HeMLN) from the previous one. Of course, hybrid MLNs (or HyMLN) are also possible.

\subsection{Structure-Preservation For Semantics}
\label{sec:structure-preservation}
For a simple graph, a community preserves its structure in terms of node/edge labels and relationships. Preserving the structure of a community of a MLN (especially HeMLN) entails preserving not only their multilayer network structure but also node/edge types, labels, and importantly inter-layer relationships. In other words, each community  should be a heterogeneous MLN in its own right!
Current approaches, such as using the multilayer network as a whole~\cite{Wilson:2017:CEM:3122009.3208030}, type-independent~\cite{LayerAggDomenicoNAL14} and projection-based~\cite{Berenstein2016, sun2013mining}, do not accomplish this as they aggregate layers into a single graph in different ways. Importantly, aggregation approaches are likely to result in some information loss~\cite{MultiLayerSurveyKivelaABGMP13} or distortion of properties~\cite{MultiLayerSurveyKivelaABGMP13} or hide the effect of different entity types and/or different intra- or inter-layer relationship combinations as elaborated in~\cite{DeDomenico201318469}. Structure-preservation avoids these and further facilitates drill-down of each community for detailed analysis.

\subsection{Decoupling Approach For Efficiency}
Decoupling as proposed in this paper is the equivalent of ``divide and conquer" for MLNs. Modeling a data set as a MLN \textit{and} computing on the whole MLN has not addressed efficiency and scalability issues. As with divide and conquer, decoupling requires partitioning (which comes from the MLN structure) and a way to compose partial (or intermediate) results. This paper adapts the bipartite graph match as the composition function (referred to as $\Theta$, with extensions) leading to efficient community detection on MLNs. Decoupling approach has also been shown to be beneficial for community detection on HoMLNs~\cite{ICCS/SantraBC17} and other computations, such as hubs, in ~\cite{ICDMW/SantraBC17}.

The contributions of this paper are:
 
\begin{itemize}
\item Definition of structure-preserving k-community for a MLN and specifically for a HeMLN, and a representation (Section~\ref{sec:hemln-community}),
\item Identification of a composition function and formalizing decoupling-based approach for k-community detection with an algorithm (Section~\ref{sec:community-detection}),
\item Guidelines for mapping detailed analysis requirements of the data set to the above definition and result computation (Section~\ref{sec:application-and-analysis}),
\item Identification of useful weight metrics and their analysis for diversity and uniqueness(Section ~\ref{sec:customize-maxflow}), and
\item Experimental analysis using the IMDb data set to establish the validity of the proposed approach for analysis along with its efficiency and versatility (Section~\ref{sec:experiments}.)
\end{itemize}

The rest of the paper is organized as follows. 
Section~\ref{sec:related-work} discusses related work relevant to this paper.
Section~\ref{sec:problem-statement} has needed definitions and a clear problem statement for the problem addressed in this paper.
Section \ref{sec:hemln-community} highlights the need for a structure-preserving community definition and formally defines the notion of a \textit{serial k-community for a HeMLN}, its representation, and the decoupling approach. Section~\ref{sec:community-detection} details the k-community detection algorithm using the decoupling approach using the maximal flow-based bipartite match. It also discusses the representation of a k-community and its benefits.
Section~\ref{sec:application-and-analysis} describes the IMDb data set as well as the analysis requirements. It also shows mapping of an analysis to a k-community specification. Section~\ref{sec:customize-maxflow} discusses the need for weight metrics, defines them formally, and their impact on k-community result. Section~\ref{sec:experiments}  shows experimental analysis of the IMDb data set using the proposed approach and validation of efficiency and importance of structure-preservation on understanding the results. Section~\ref{sec:conclusions} concludes the paper.

\section{Related Work}
\label{sec:related-work}

As the focus of this paper is community definition and its efficient detection in HeMLNs, we present relevant work on single or simple graphs (monoplexes) and MLNs (both homogeneous and heterogeneous.) The advantages of modeling using MLNs are discussed in~\cite{Boccaletti20141, BDA/SantraB17,CommSurveyKimL15, MultiLayerSurveyKivelaABGMP13}.

\textit{Community detection} on a simple graph involves identifying groups of vertices that are more connected to each other than to other vertices in the network. Most of the  work in the literature considers \textbf{single networks or simple graphs} where this objective is translated to optimizing network parameters such as modularity ~\cite{clauset2004} or conductance ~\cite{Leskovec08}. As the combinatorial optimization of community detection is NP-complete ~\cite{Brandes03}, a large number of competitive approximation algorithms have been developed (see reviews in ~\cite{Fortunato2009, Xie2013}.) Algorithms for community detection have been developed for different types of input graphs including directed ~\cite{Yang10, Leicht08},
edge-weighted ~\cite{Berry2011}, and dynamic networks ~\cite{porterchaos13,Bansal2011}.
However, to the best of our knowledge, there is no work on  community detection that include node and edge labels, node weights as well as graphs with self-loops and multiple edges between nodes\footnote{In contrast, subgraph mining~\cite{datamine/KuramochiK05,KDD/HolderCD1994, tkde/DasC18}, querying~\cite{tkde/JayaramKLYE15,dawak/DasGC16}, and search~\cite{bigdataconf/HaoC0HBH15,pods/ShashaWG02} have used graphs with node and/or edge labels including multiple edges between nodes, cycles, and self-loops.}. Even the most popular community detection packages such as Infomap \cite{InfoMap2014} or Louvain~\cite{DBLP:Louvain}, do not take these parameters into consideration. 

Recently, community detection algorithms have been extended to \textbf{HoMLNs}. Algorithms based on matrix factorization \cite{dong2012clustering,qi2012community},
pattern mining \cite{silva2012mining,zeng2006coherent},
cluster expansion philosophy \cite{li2008scalable}, Bayesian probabilistic models \cite{xu2012model}, regression \cite{cai2005mining}, spectral optimization of the modularity function based on the supra-adjacency representation \cite{zhang2017modularity} and a significance based score that quantifies the connectivity of an observed vertex-layer set
through comparison with a fixed degree random graph model \cite{wilson2017community} have been developed. However, all these approaches \textit{analyze a MLN either by aggregating all (or a subset of) layers (straightforward for a HoMLN using Boolean and other operators) or by considering the entire multiplex as a whole}. Recently a decoupling-based approaches for detecting communities~\cite{ICCS/SantraBC17} and centrality~\cite{ICDMW/SantraBC17} in HoMLN have been proposed. 
They reduce the exhaustive analysis complexity from O($2^N$) (for all possible subsets of layers) to linear complexity for \textit{N} layers by composing combinations of them. The goal of this paper is to do the same for HeMLNs.

Majority of the work on analyzing HeMLN (reviewed in \cite{shi2017survey,sun2013mining}) focuses on developing meta-path based techniques for determining the similarity of objects~\cite{wang2016relsim}, classification of objects~\cite{wang2016text}, predicting the missing links~\cite{zhang2015organizational}, ranking/co-ranking~\cite{shi2016constrained} and recommendations~\cite{shi2015semantic}. An important aspect to be noted here is that most of them do not consider the intra-layer relationships and concentrate mainly on the bipartite graph formed by the inter-layer edges. 

The type-independent~\cite{LayerAggDomenicoNAL14} and projection-based~\cite{Berenstein2016, sun2013mining} approaches used for HeMLNs do not preserve the structure or types and labels (of nodes and edges) of the community. The type independent approach collapses  all layers into a single graph keeping \textit{all} nodes and edges (including inter-layer edges) sans their types and labels. Similarly, as the name suggests, the projection-based approach projects the nodes of one layer onto another layer and uses the layer neighbor and inter-layer edges to collapse the two layers into a single graph with a single entity type instead of two.
The presence of different sets of entities in each layer and the presence of intra-layer edges makes structure-preserving definition more challenging for HeMLNs. A few existing works have proposed techniques for generating clusters of entities \cite{sun2009rankclus, sun2009ranking,melamed2014community}, but they have only considered the inter-layer links and not the networks themselves. Thus, the \textit{combined effect of layer communities, entity types, intra- and inter-layer relationships (types) have not been included in defining a community in a HeMLN}.  This paper hopes to fill that gap.

\section{Definitions}
\label{sec:problem-statement}

A {\bf graph} $G$ is an ordered pair $(V, E)$, where $V$ is a set of vertices and $E$ is a set of edges. An edge $(v,u)$ is a 2-element subset of the set $V$. The two vertices that form an edge are said to be adjacent or neighbors of each other. In this paper we only consider graphs that are undirected (the  vertices in the edge are unordered) and simple (there are no self-loops or multiple edges.) Although this type of graph does not include labels, we note that our representation and analysis can be applied to labeled graphs.

A {\bf multilayer network}, $MLN (G, X)$, is defined by two sets of single networks. The set $G = \{G_1, G_2, \ldots, G_N\}$ contains general graphs for N layers as defined above, where $G_i (V_i, E_i)$ is defined by a set of vertices, $V_i$ and a set of edges, $E_i$. An edge $e(v,u) \in E_i$, connects vertices $v$ and $u$, where $v,u\in V_i$. The set $X =\{X_{1,2}, X_{1,3}, \ldots, X_{N-1,N}\}$ consists of bipartite graphs. Each graph $X_{i,j} (V_i, V_j, L_{i,j})$ is defined by two sets of vertices $V_i$ and $V_j$, and a set of edges (or links) $L_{i,j}$, such that for every link $l(a,b) \in L_{i,j}$,  $a\in V_i$ and $b \in V_j$, where $V_i$ ($V_j$) is the vertex set of graph $G_i$ ($G_j$.)

{\em To summarize}, the set $G$, represents the constituent graphs of layers of a MLN and the set $X$ represents the bipartite graphs that correspond to the edges \textit{between} pairs of layers (or inter-layer edges.) The two sets of nodes for any bipartite graph from $X$ come from two layers and if labels are assigned, they would have different node labels. For a Heterogeneous MLN (HeMLN), $X$ is explicitly specified.
Without loss of generality, we assume unique numbers for nodes across layers and disjoint sets of nodes across layers\footnote{Heterogeneous MLNs are also defined with overlapping nodes across layers (see \cite{MultiLayerSurveyKivelaABGMP13}.) We do not consider that in this paper.}.

We  propose a {\bf decoupling approach for community detection}. Our algorithm is defined for combining communities from two layers of a HeMLN (using a composition function) which is extended to $k$ layers (by applying pair-wise composition repeatedly.) We define a \textit{serial k-community} to be a multilayer community where communities from $k$ distinct connected layers of a HeMLN are combined in a specified order.

\begin{figure}[h]
   \centering \includegraphics[width=\linewidth]{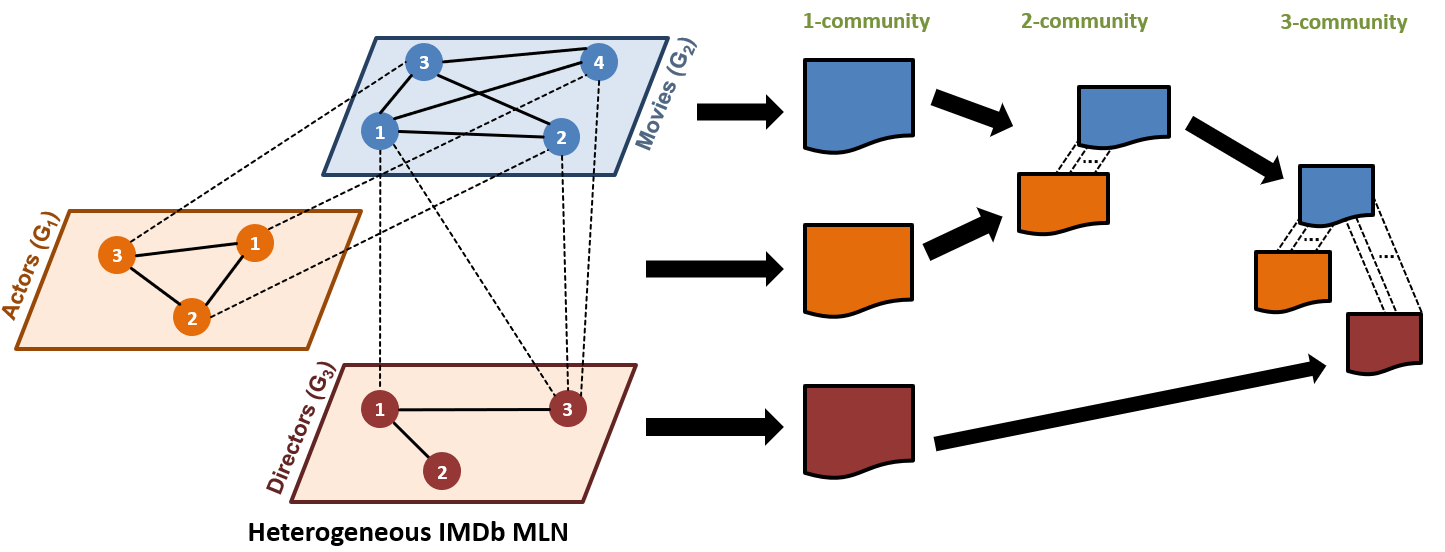}
    \caption{Illustration of decoupling approach for computing a 3-community (($G_2$ $\Theta_{2,1}$ $G_1$)  $\Theta_{2,3}~G_3$); $\omega_e$\protect\footnotemark}
   \label{fig:decoupling}
\end{figure}

\footnotetext{Technically, this should be expressed as (($\Psi$($G_2$) $\Theta_{2,1}$ $\Psi$($G_1$))  $\Theta_{2,3}$ $\Psi$($G_3$).) However, we drop $\Psi$ for simplicity. In fact, $\Theta$ with its subscripts is sufficient for our purpose due to pre-defined precedence (left-to-right) of $\Theta$. We retain G for clarity of the expression. $\omega_e$ is a weight metric discussed in Section~\ref{sec:customize-maxflow}.}

Network decoupling approach for finding communities in MLNs is as follows;

 {\bf(i)} First use the function $\Psi$ (here community detection) to find communities in each of the layers (networks) individually and consider each community to be a meta node,

 {\bf(ii)} for any two chosen layers, construct a bipartite graph using their communities as meta nodes and creating meta edges that connect the meta nodes (using an element of X) and assign weights ($\omega$), and

 {\bf(iii)} finally combine the partial results from each layer (represented as meta nodes of the bipartite graph above), using a function $\Theta$ (here \textit{maximal network flow matching}.)

 Although in this paper we focus on community detection, the functions $\Psi$, $\omega$ and $\Theta$ can be suitably modified for other network analysis, such as centrality detection, as well.

 Figure~\ref{fig:decoupling} illustrates the decoupling approach for specifying and computing a community of a larger size from partial results. We illustrate how a set consisting of distinct communities from a layer is used for computing a 2-community (for 2 layers) and further a 3-community (for 3 layers) using partial results. 1-community is the set of communities generated for a layer (simple graph.) As shown in the figure, 1-community from individual layers are used to compose progressively larger k-community.

We can thus succinctly define our problem statement as follows;
\textit{For a given data set with $F$ different features and $T$ entity types and a set of analysis objectives $A$, create an appropriate HeMLN and determine the appropriate triad of $\Psi$,  $\Theta$, and $\omega$, and the order in which to combine the layers to obtain the serial $k$-community for \textit{each objective}, where $k \le F$.} 
\section{What is a Community in a H\MakeLowercase{e}MLN?}
\label{sec:hemln-community}

As discussed in Section~\ref{sec:related-work}, existing approaches for detecting communities for HeMLNs include either i) combining all (or a subset of) layers and applying single graph community detection algorithms or ii) projecting a layer on to a connected layer as a means of aggregating the two layers and applying single graph community detection. Both ignore the different types of entities, and are likely to result in some information loss, distortion of properties or hide the effect of different intra- or inter-layer relationships~\cite{MultiLayerSurveyKivelaABGMP13,DeDomenico201318469}. However, these approaches are used because they can leverage the availability of algorithms available for simple graphs.

\begin{table}[h]
    \renewcommand{\arraystretch}{1.3}
    \centering
    \begin{tabular}{|c||c|}
            \hline
            $G_i (V_i, E_i)$ & Simple Graph for layer \textit{i} \\
            \hline
            $X_{i,j} (V_i,V_j,L_{i,j})$ & Bipartite graph of layers \textit{i} and \textit{j} \\
            \hline
            $MLN (G, X)$ & Multilayer Network of layer graphs (set \textit{G}) and Bipartite graphs (set \textit{X})\\
            \hline
            $\Psi$ & Analysis function for $G_i$ (community)\\
            \hline
            $\Theta_{i,j}$ & Bipartite graph match function \\
            \hline
            $CBG_{i,j}$ & Community bipartite graph for $G_i$ and $G_j$ \\
            \hline
            $U_i$ & Meta nodes for layer \textit{i} 1-community \\
            \hline
            $L'_{i,j}$ & Meta edges between $U_i$ and $U_j$ \\
            \hline
            $c_i^{m}$ & $m^{th}$ community of $G_i$ \\
            \hline
            $v_i^{c^{m}}$, $e_i^{c^{m}}$ & Vertices and Edges in community $c_i^{m}$\\
            \hline
            $H_i^m$ & Hubs in $c_i^{m}$ \\
            \hline
            $H_{i,j}^{m,n}$ & Hubs in $c_i^{m}$ connected to $c_j^{n}$\\
            \hline
            $x_{i,j}$ & \{Expanded(meta edge $<c_i^{m}$, ~$c_j^{n}>$\}\\
            \hline
            $0$ and $\phi$ & null community id and empty $x_{i,j}$\\
            \hline
            $\omega_e$, $\omega_d$, $\omega_h$ & Weight metrics for meta edges\\
            \hline
        \end{tabular}
            \caption{List of Notations used in this paper}
    \label{table:notations}

\end{table}

We now present the case for creating structure preserving communities.

As an example, for the HeMLM shown in Figure~\ref{fig:decoupling} for the IMDb data set, consider the analysis \textit{``Find group of directors who direct actors who act together?} Note that the actor and director layers can only compute groups of directors (who direct similar genre) and groups of actors (who act together.) The connection between directors and actors only come from inter-layer edges. It is only by preserving the structure of both the communities in actor and director as well as the inter-layer edges, can we answer this question.

Clearly, multiple relationships can exist in such a collection of layers, such as co-acting, similar genres and who-directs-whom. An analysis requirement may also want to use \textit{preferences} for community interactions. As an example, one may be interested in groups (or communities) where the \textit{most important} actors and directors interact.

The community definition and detection research in the literature for homogeneous MLNs~\cite{BDA/SantraB17,ICCS/SantraBC17} are not applicable to HeMLN as each layer has \textit{different sets and types of entities} with  \textit{inter-layer edges} between them.
We propose a general intuitive community definition for HeMLN consisting of any number of layers and arbitrary inter-layer connections.
Importantly, \textit{this definition of community for a HeMLN needs to be structure-preserving and computationally efficient}.

Table~\ref{table:notations} lists all notations used in the paper and their meaning for quick reference

\subsection{Formal Definition of Community in a MLN\protect\footnote{All the definitions given below are applicable to both a HoMLN and a HeMLN, as well as for disjoint or overlapping communities. However, we are focusing only on disjoint communities in HeMLNs in this paper.}}
\label{sec:definition}

A \textbf{\textit{1-community}} is a set of communities of the simple graph corresponding to  a layer.

A \textbf{community bipartite graph\footnote{We defined the set X of bipartite graphs between layers of HeMLN in Section~\ref{sec:problem-statement}. This is a different bipartite graph between two sets whose nodes (termed meta nodes) correspond to communities from two distinct layers. A single bipartite edge (termed meta edge) is drawn between distinct meta node pairs as defined.}}
\textbf{\textit{CBG$_{i,j}$($U_i$, $U_j$, $L'_{i,j}$)}}
is defined between two disjoint and independent sets $U_i$ and $U_j$. An element of $U_i$ ($U_j$) is a community id from $G_i$ ($G_j$) and is represented as a single meta node in $U_i$ ($U_j$.) $L'_{i,j}$ is the set of meta edges between the nodes of $U_i$ and $U_j$ (or bipartite graph edges.) For any two meta nodes, a \textit{single edge} is included in $L'_{i,j}$, if there is \textit{an inter-layer edge} between any pair of nodes from the corresponding communities (acting as meta nodes in CBG) in layers $G_i$ and $G_j$. Note that there may be many inter-layer edges between the communities from the two layers. Also note that $U_i$ ($U_j$) need not include all community ids of $G_i$ ($G_j$.) The strength (or weight) component of the meta edges is elaborated in Section \ref{sec:customize-maxflow}.

A \textbf{\textit{serial 2-community}} is defined on the community bipartite graph \textit{CBG$_{i,j}$($U_i$, $U_j$, $L'_{i,j}$)} corresponding to layers $G_i$ and $G_j$\footnote{All 1-community elements from $G_i$ ($G_j$) need not be in $U_i$ ($U_j$.)}. A 2-community is a set of tuples each with a pair of elements $<c_i^m, c_j^n>$, where $c_i^m \in U_i$ and $c_j^n \in U_j$, that satisfy the \textit{maximal network flow property} (composition function $\Theta$ defined in Section~\ref{sec:problem-statement}) for the bipartite graph of $U_i$ and $U_j$, along with the set of inter-layer edges between them.
The number of tuples in the 2-community is bound by the $Min(|U_i|, |U_j|)$.

A \textbf{\textit{serial k-community}}\footnote{We have used the \textbf{``serial"} prefix for defining a k-community to emphasize the order used (but can be arbitrary) in its specification as a k-community corresponds to a connected subgraph of k layers. Our definition assumes left-to-right precedence for the composition function $\Theta$. It is possible to define a k-community with explicit precedence specification for $\Theta$.
Also, other definitions are possible that may be order agnostic. Finally, we drop the repetitive ``serial" prefix henceforth  as we only refer to a serial k-community in the rest of the paper.}  for \textit{k} layers of a HeMLN is defined as applying the \textit{serial 2-community} definition recursively to compose a k-community.
The base case corresponds to applying the definition of 2-community for some two layers. The recursive case corresponds to applying  2-community composition for a  t-community.

For each recursive  step, there are two cases: i)  for the 2-community under consideration, the $U_i$  is from a layer $G_i$ \textit{already in the t-community} and the $U_j$  is from a \textit{new layer $G_j$}. This bipartite graph match is said to \textbf{extend} a t-community (t $<$ k) to a \textit{(t+1)-community},
or ii) for the 2-community under consideration,  \textbf{both} $U_i$ ($U_j$)  from layers $G_i$ ($G_J$) are \textit{already in the t-community}.
This bipartite graph match is said to \textbf{update} a t-community (t $<$ k), not extend it.

In both cases i) and ii) above, a number of outcomes are possible. Either a meta node from $U_i$ a) matches a meta node in $U_j$ resulting in a \textbf{consistent match}, or b) does not match a meta node in $U_j$ resulting in a \textbf{no match}, or c) matches a node in $U_j$ that is not consistent with a previous match termed \textbf{inconsistent match}.

Structure preservation is accomplished by retaining for each tuple of t-community, as appropriate, either a matching community id (or 0 for null id) and/or
$x_{i,j}$ (or $\phi$ for empty set) representing inter-layer edges corresponding to the meta edge between the meta nodes (termed \textbf{expanded (meta edge).)} The \textit{extend} and \textit{update} carried out for each of the outcomes on the representation is listed in Table~\ref{table:algo-refer}.

\begin{table}[h!t]
\centering
    \begin{tabular}{|l|c|}
            \hline
                \textbf{($G_{left}$, $G_{right}$) outcome} & \textbf{Effect on tuple \textit{t} } \\
            \hline
            \hline
            \multicolumn{2}{|c|}{\texttt{\textcolor{blue}{case (i) - one processed and one new layer}}} \\
            \hline
            a) \texttt{\textcolor{blue}{consistent match}} & \textbf{Extend} \textit{t} with paired community id \textbf{and} $x_{i,j}$\\
            \hline
            b) \texttt{\textcolor{blue}{no match}} & \textbf{Extend} \textit{t} with 0 and $\phi$\\
            \hline
            \hline
            \multicolumn{2}{|c|}{\texttt{\textcolor{blue}{case (ii) - both are processed layers}}} \\
            \hline
            a) \texttt{\textcolor{blue}{consistent match}} & \textbf{Update} \textit{t} \textbf{only with} $x$\\
            \hline
            b) \texttt{\textcolor{blue}{no match}} & \textbf{Update} \textit{t} \textbf{only with} $\phi$\\
            \hline
            c) \texttt{\textcolor{blue}{inconsistent match}} & \textbf{Update} \textit{t} \textbf{only with} $\phi$\\
            \hline

        \end{tabular}
            \caption{Possible cases and outcomes for a maximal network flow match (See Algorithm \ref{alg:k-community})}
    \label{table:algo-refer}
\end{table}

Note that when a k-community is defined on a connected subgraph \textit{sg} (consisting of k layers of a HeMLN as \textit{single nodes} and each bipartite edge set as a \textit{single edge}), for any recursive step, if \textbf{only} case i) is applicable, then it corresponds to a k-community of \textit{a} spanning tree of \textit{sg}. On the other hand, use of case ii) corresponds to a k-community definition in which one or more cycle edges  of \textit{sg} participate. For both cases, maximal network flow outcome may result in 3 outcomes:   a \textit{consistent match}, \textit{no match}, or an \textit{inconsistent match} as shown in Table~\ref{table:algo-refer}.

\subsection{Characteristics of k-community}

The above definition when applied to a specification (such as the one shown in Figure~\ref{fig:sequence}) generates \textit{progressively strong coupling between layers} (due to left-to-right precedence of $\Theta$) using maximal network flow. \textit{Thus, our definition of a k-community is characterized by dense connectivity within the layer (community definition) and strong coupling across layers (maximal network flow definition.)} Note that the traditional methods of finding communities in HeMLN, such as aggregation or projection can be expressed as special cases of our structure-preserving definition.

 \begin{figure}[ht]
   \centering \includegraphics[width=\linewidth]{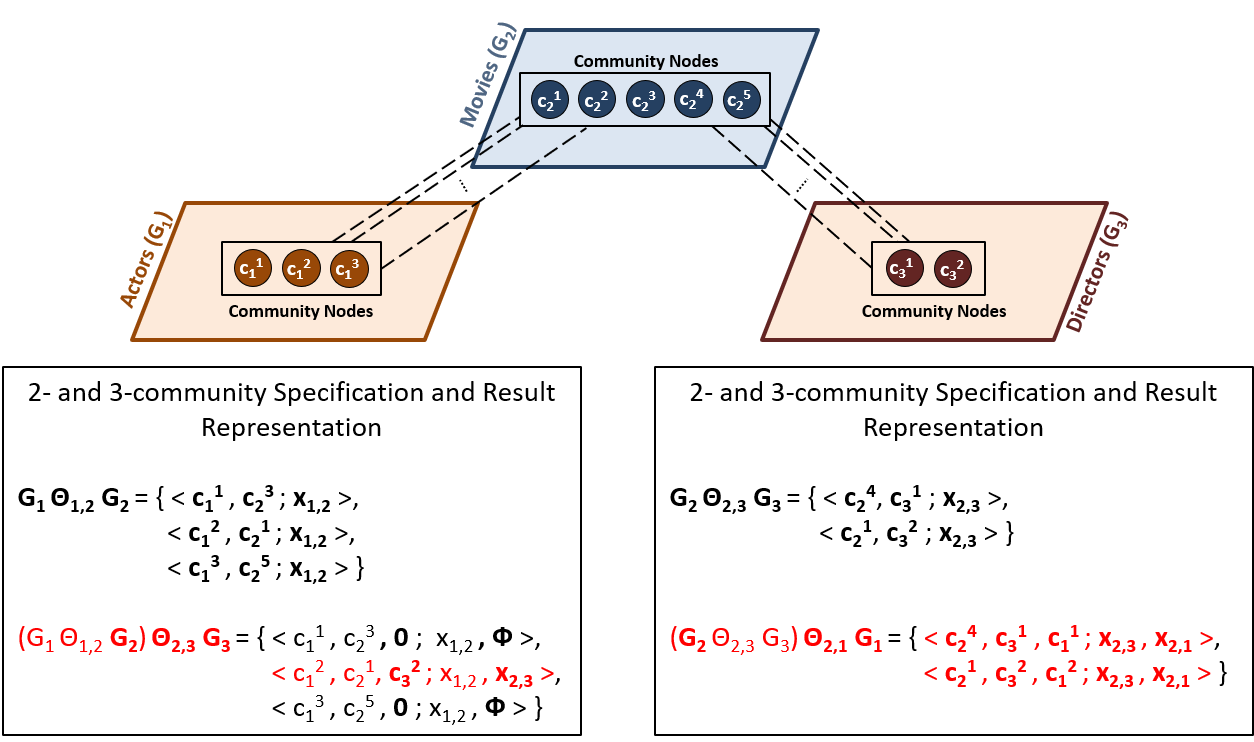}
   \caption{{Illustration exemplifying order dependence on a k-community}}
   \label{fig:sequence}
\vspace{-15pt}
\end{figure}

{\bf Space of Analysis Alternatives: } Given a connected HeMLN of k layers, the number of possible k-community is quite large. If we consider the HeMLN layers to be single nodes and bipartite graph edge sets to be single edges, then it is a function of the number of unique connected subgraphs of different sizes and the number of possible orderings for each such connected subgraph (including cycles.) With the inclusion of 3 weight metrics (see below), it gets even larger.
It is important to understand that each subgraph of a given size (equal to the number of edges in the connected subgraph) along with the ordering represents a \textit{different} analysis of the data set and provides a different perspective thereby supporting a large space of analysis alternatives. Finally, the composition function $\Theta$ defined above (maximal network flow of a bipartite graph) is commutative  follows from the maximal network flow), but not associative\footnote{This is due to the use of a subset of meta nodes rather than the entire 1-community during any recursive step.}. Hence, for each k-community, the \textit{order in which a k-community} is defined has a bearing on the result (semantics) obtained for the k-community. In fact, the ordering is important as it differentiates one analysis from the other even for the same set of layers  and inter-layer connections as elaborated in Section~\ref{sec:application-and-analysis}. Figure~\ref{fig:sequence} shows clearly two 3-community for the same layers which are quite different! For the Figure~\ref{fig:sequence}, assuming inter-layer edges between layers A and D (forming a cycle), possible 1-community specifications are 3, possible  2-community specifications are 3, the number of acyclic 3-community specifications are 6, and the number of cyclic 3-community specifications are 6 for a total of 18 possible analysis.

{\bf Bounds on the size of k-community.} The number of elements in the  k-community  is bounded by the number of consistent matches (or pairs) obtained for the base case.
Bipartite match in every recursive step is bound by the smaller of the two sets $U_i$ and $U_j$. Also, each element of a k-community set can be total or partial. A \textbf{partial k-community element} has \textit{\textbf{at least one} $\phi$} as part of the tuple. Otherwise, it is a \textbf{total k-community element}.
If a cyclic k-community shows a \textbf{total}  element, it reflects a much stronger coupling as it includes all inter-layer edges for those communities in the cycle (as is the case of M2-A2-D4-M4 in Figure~\ref{fig:A6-madm}.) A \textbf{partial} k-community  element, on the other hand, for both acyclic and cyclic cases indicates strong coupling only among \textit{a subset of layers}.

{\bf Need For Edge Weights:} Finally, in a traditional bipartite graph (used for dating, hiring etc.), each node is a simple node. In our case, each meta node is a community representing a group of entities with additional characteristics. Hence, the weights associated with the edges of the community bipartite graph need to reflect the participating community characteristics which has a significant bearing on the strength of the coupling. We discuss a number of alternatives for weights (termed weight metrics $\omega$) in Section~\ref{sec:customize-maxflow}, derived from real-world scenarios.

The above definition can be used for both disjoint or overlapping 1-community. We believe that this definition can be applied to a homogeneous MLN as well by defining the weight metrics appropriately as there are significant differences in the inter-layer edge characteristics and node types. \textit{Most importantly, unlike current alternatives in the literature for community of a MLN, there is no information loss or distortion or hiding the effect of different entity types or relationships in our definition. } 
\section{\MakeLowercase{k}-Community Detection}
\label{sec:community-detection}

In this section, we first present a linear specification of a k-community and also a linear representation for the result. We indicate how the result representation preserves the structure along with all related information. Then we present an algorithm using the weight metric given for maximal network flow. We discuss a number of meaningful ways in which we can enhance the strength of the coupling for maximal network flow approach by providing alternative weight metrics based on participating community graph characteristics.
This allows us to adapt an available algorithm to factor in the semantics of the analysis.

\subsection{k-Community Representation}
\label{sec:rep}

Note that a k-community is represented as a set of tuples. Each tuple represents a distinct element of a k-community and includes an ordering of k community ids as items  and at least an ordered (k-1)  expanded(meta edge) (i.e., $x_{i,j}$.) This representation completely preserves the MLN structure as well as the ability to reconstruct node/edge types and labels.

Linearization of a HeMLN structure is done using an order of specification which is also used for computation.  Although a k-community need to be specified as an expression involving $\Psi$ and $\Theta$, as indicated earlier, we drop $\Psi$ for clarity. For the layers shown in Figure~\ref{fig:sequence}, a 2-community specification could be ($G_1 ~\Theta_{1,2}~ G_2$). A 3-community computation shown is for the specification (($G_2$ $\Theta_{2,3}$ $G_3$)  $\Theta_{2,1}$ $G_1)$. We can drop the parentheses as the precedence of $\Theta$ is assumed. However, we need the subscripts for $\Theta$ to disambiguate a k-community specification  when a composition is done on the layers already used. 
If the layers $G_1 ~ and ~G_3$ were also connected with inter-layer edges in Figure~\ref{fig:sequence}, a 3-community involving a cycle can be specified as $G_1$ $\Theta_{1,2}$ $G_2$  $\Theta_{2,3}$ $G_3 ~\Theta_{3,1}~ G_1$.  

The result of a k-community is represented as a set of tuples where each tuple corresponds to a distinct element of k-community. Each tuple of a k-community has 2 distinct components. The first component of each tuple is an ordering of community ids (as items) from a distinct layer in the k-community specification, separated by a comma. A semicolon separates the second component from the first. The second component of each tuple is again a comma-separated ordering of (k-1) items which are $x_{i,j}$ (with each x having a different pair of subscripts.)  Communities for $x$ are uniquely identifiable from the subscripts. Figure~\ref{fig:sequence} shows a number of 2- and 3-community results for the corresponding specifications.
For the acyclic 3-community specification $G_1$ $\Theta_{1,2}$ $G_2$ $\Theta_{2,3}$ $G_3$, the element $<$ $c_1^2$, $c_2^1$, $c_3^2$ \textbf{;} $x_{1,2}$, $x_{2,3}$ $>$ is the only total element as it does not include any $\phi$, whereas, the other two elements, $<$ $c_1^1$, $c_2^3$, 0 \textbf{;} $x_{1,2}$, $\phi$ $>$ and $<$ $c_1^3$, $c_2^5$, 0 \textbf{;} $x_{1,2}$, $\phi$ $>$, are partial as both include one $\phi$. 
If $G_1$ and $G_3$ were connected, then for the 3-community specification $G_1$ $\Theta_{1,2}$ $G_2$ $\Theta_{2,3}$ $G_3$ $\Theta_{3,1}$ $G_1$ (involving a cycle), the result shown in figure changes to \{
 $<$ $c_1^1$, $c_2^3$, 0 \textbf{;} $x_{1,2}$, $\phi$ $>$, $<$ $c_1^2$, $c_2^1$, $c_3^2$ \textbf{;} $x_{1,2}$, $x_{2,3}$, $x_{3,1}$ $>$, $<$ $c_1^3$, $c_2^5$, 0 \textbf{;} $x_{1,2}$, $\phi$ $>$
\}. Note that the number of communities in each tuple is k (3 here) and the number of inter-layer edge sets is at least (k-1). It is exactly (k-1) if the k-community is for a spanning tree (i.e., an acyclic connected graph) and more depending upon the number of cyclic edges in specification (here it is 3 as one cyclic edge is included.) To generalize, an element of \textit{\textbf{k-community}} for an arbitrary specification \\
$G_{n1}$ $\Theta_{n1,n2}$ $G_{n2}$ $\Theta_{n2,n3}$ $G_{3}$ ... $\Theta_{n1,nk}$ $G_{k}$ \\ will be represented as \\ $<$ $c_{n1}^{m1}$, $c_{n2}^{m2}$, ..., $c_{nk}^{mk}$ \textbf{;} $x_{n1,n2}$, $x_{n2,n3}$, ..., $x_{ni,nk}$ $>$, where some c's may be 0 and some x's may be $\phi$.

\subsubsection{Structure-Preserving Representation Benefits}

The representation used has a number of benefits from an analysis perspective.

\begin{enumerate}
    \item Each element of s k-community can be further analyzed individually as the tuple contains all the information to reconstruct the HeMLN and drill down for details. Each community can also be displayed.
    \item Total and partial elements of k-community, each provide important information about the result characteristics.  A partial one shows a weak coupling of the complete community whereas a total element indicates strong coupling. This can be inferred by processing the result set.
    \item The resulting set can be ranked in several ways based on community and other characteristics. For example, they can be ranked based on community size or density (or any other feature) as well as significance of the layer.
    \item Finally, results from previous specifications can be re-used if the same sub-expression needs to be evaluated again.
    
\end{enumerate}

\subsection{k-Community Detection Algorithm}
\label{sec:flow}

\begin{algorithm}[h]
\caption{k-community Detection Algorithm}
\label{alg:k-community}
\begin{algorithmic}[1]
   \REQUIRE- \\
   \textbf{INPUT:}  HeMLN, ($G_{n1}$ $\Theta_{n1,n2}$ $G_{n2}$ ... $\Theta_{ni,nk}$ $G_{nk}$), and a weight metric (wm). \\
   \textbf{OUTPUT:} Set of tuples with two components \\
    \texttt{//input processed from left to right} \\
    \STATE \textbf{Initialize:} k = 2, $U_i$ = $\phi$, $U_j$ = $\phi$, $L$ = $\phi$\\  
    \textit{result} $\gets$ initialize(2-community($G_{n1}$,$G_{n2}$),HeMLN,wm) \\
    \textit{left} $\gets$ next\_left\_subscript ($\Theta$) \\
    \textit{right} $\gets$ next\_right\_subscript($\Theta$) \\
    \WHILE{\textit{left} $\neq$ null \&\& \textit{right} $\neq$ null}
        \STATE $U_i$ $\gets$ subset of 1-community($G_{left}$) \\
        \STATE $U_j$ $\gets$ subset of     1-community($G_{right}$) \\ 
        \texttt{// subsets (3,4) if layer has been processed} \\
    
        \STATE \textit{MP} $\gets$ max\_flow\_pairs($U_i$, $U_j$, HeMLN, wm) \\
        \texttt{// a set of pairs $<c_{left}^p$, $c_{right}^q>$} \\
        \FOR{ \textbf{each} tuple \textit{t} $\in$ \textit{result} }
            \IF {\textit{both $c_{left}^x ~and~ c_{right}^y$} are part of \textit{t} and $\in$ MP \texttt{\textcolor{blue}{[case ii (processed layer): consistent match]}}}
                \STATE Update \textit{t} with  ($x_{left, right}$) 
            \ELSIF {$c_{left}^x$ is  part of \textit{t} and $\in$ MP and $G_{right}$ layer has been processed \texttt{\textcolor{blue}{[case ii (processed layer): no and inconsistent match]}} }
                \STATE Update \textit{t} with $\phi$ \\
            \ELSIF {$c_{left}^x$ is part of \textit{t} and $\in$ MP  \texttt{\textcolor{blue}{[case i (new layer): consistent match]}}}
                \STATE Extend \textit{t} with paired $c_{right}^y$ $\in$ MP and $x_{left, right}$   \\
                \STATE k = k + 1
            \ELSIF {$c_{left}^x$ is part of \textit{t} and $\notin$ MP \texttt{\textcolor{blue}{[case i (new layer): no match]}}}
                \STATE Extend \textit{t} with 0 (community id) and $\phi$ \\
                \STATE k = k + 1
            \ENDIF
        \ENDFOR \\
        \textit{left} $\gets$ next\_left\_subscript ($\Theta$) or null \\
        \textit{right} $\gets$ next\_right\_subscript($\Theta$) or null
        \ENDWHILE
     \end{algorithmic}
\end{algorithm}

Algorithm~\ref{alg:k-community} is an iterative algorithm that accepts a linearized specification of a k-community and computes the result as described earlier. The input is an \textit{ordering of layers}, \textit{composition function indicating the community bipartite graphs to be used} and the type of weight to be used (elaborated below.) 
The output is a \textit{set} whose \textit{elements are tuples corresponding to  distinct, single elements of k-community} for that specification. The size (i.e., number of tuples) of this set is bound by the base case. The layers for any 2-community bipartite graph composition are identifiable from the input specification. 

The algorithm iterates until there are no more compositions to be applied. The number of iterations is equal to the number of $\Theta$ in the input (including $\Theta$ for the base case.)  For each layer, we assume that its 1-community has been computed.

The bipartite graph for the base case and for each iteration is constructed for the participating layers (either one is new or both are from the t-community for some t) and maximal network flow algorithm is applied. The result is used to either extend or update the tuples of the t-community for all the cases as described in Table~\ref{table:algo-refer}.
Note that the k-community size\textbf{ k is incremented} only when a \textit{new layer is composed (case i).)} For case ii) (cyclic k-community) \textbf{k is not incremented} when \textit{both layers are part of the t-community}.
 When the algorithm terminates, we will have the set of tuples each corresponding to a single, distinct element of k-community for the given specification.

It is not necessary to use the \textit{same weight metric} for all bipartite graph matches while computing a k-community. It is possible to choose different metrics for each bipartite match to adhere to the analysis semantics. Algorithm~\ref{alg:k-community}, as given, does not support this.

Figure~\ref{fig:sequence} illustrates examples of 2- and 3-community computed using the above algorithm (both specification and results.)  The results shown using the representation. The figure also shows the two components of each tuple comprising of community id as items in the first component and the set of expanded(meta edge) or $x$ as items for the second component. Further, it shows how the result of a 2-community is extended to form a 3-community.  It also demonstrates the importance of order in which a k-community is defined. Further, it shows that the same result need be not obtained for different ordering (e.g., 3-community) \textit{on the same layers}.  

Although this paper only discusses the serial k-community, we want to emphasize that the proposed approach -- both specification and computation -- can be easily extended and generalized to define variants of k-community. For example, by indicating  precedence for applying 2-community bipartite match, a non-serial k-community can be easily defined and computed by the above algorithm with minor changes. For example, the k-community specification 
$(((G_3$ $\Theta_{3,2}$ $G_2$)  $\Theta_{2,1}$ $G_1) ~\Theta_{2,3}~ G_3)$ could be specified alternatively as a \textit{non-serial k-community} as
$((G_3$ $\Theta_{3,2}$ $G_2$)  $\Theta_{3,2}$ ($G_1 ~\Theta_{1,2}~ G_2))$ which has a totally different semantics from an analysis perspective. 

Finally, the algorithm uses maximal network flow algorithm with weights for meta edges that are customized for any MLN k-community. From a  broader analysis perspective, we have identified several important characteristics to formulate edge weights. We discuss only three of them  in this paper (in Section~\ref{sec:customize-maxflow}) due to space constraints, which are (all between participating communities in the community bipartite graph): number of edges (denoted by $\omega_e$), density (denoted by $\omega_d$), and hub participation (denoted by $\omega_h$.)
\section{Analysis of IMD\MakeLowercase{b} Data Set}
\label{sec:application-and-analysis}

In order to demonstrate the need and utility of structure-preserving k-community and identify a set of metrics that are domain-independent and applicable for specific analysis needs, we introduce the data set and analysis requirements. This will clarify the usage of the three weight metrics for maximal network flow algorithm. The same data set is used for experimental analysis.\\

\noindent \textbf{IMDb data set: }The IMDb data set captures movies, TV episodes, actor, directors and other related information, such as rating. This is a large data set consisting of movie and TV episode data from their beginnings. This data set can be modeled and analyzed in multiple ways.
This data set has been chosen for its versatility in that it can be modeled using HoMLN as well as HeMLN based on analysis requirements.\\

\noindent \textbf{General Analysis Requirements: }\textit{The analysis goal for this data set is to analyze the three important groups (entity types)  in this data set: actors, directors, and movies based on the information available (attributes.) Given three entity types, we are interested in analyzing directs-actor relationship, acts-in-a-movie with a specific-rating relationship, and directs-movie with a specific review rating relationship.}
Below, the above analysis is expressed as detailed requirements to facilitate mapping to a k-community. We also indicate the specification for each detailed analysis including the order of composition followed by metric to be used. The rationale for these mappings are discussed in the remainder of the paper. \\

\noindent \textbf{Detailed Analysis Requirements: }

\begin{enumerate}[label={A(\arabic*)}]
    \item Find co-actor groups that have \textit{maximum interaction} with  director groups who have directed similar genres?

    2-community: A $\Theta_{A,D}$ D; $\omega_e$; order does not matter.
    \label{list:maximum-interaction}

    \item Identify the \textit{strongly connected} co-actor groups where \textit{most of the actors have worked with most of the directors} who have directed similar genres
    
    2-community: A $\Theta_{A,D}$ D; $\omega_d$; order does not matter.
    \label{list:density}
    
    \item Identify \textit{versatile} director groups who work with \textit{most sought after} actors among co-actors

    2-community: A $\Theta_{A,D}$ D; $\omega_h$; order does not matter.
    \label{list:hub-participation}

    \item For the group of directors (who direct similar genres) having \textit{maximum interaction} with members of co-actor groups, identify the \textit{most popular rating} for the movies they direct? 

    Acyclic 3-community: A $\Theta_{A,D}$ D $\Theta_{D,M}$ M; $\omega_e$
    \label{list:We-acyclic3-adm}

    \item For the \textit{most popular} actor groups from each movie rating class, which are the director groups with which they have \textit{maximum interaction}?

    Acyclic 3-community: M $\Theta_{M,A}$ A $\Theta_{A,D}$ D; $\omega_e$
    \label{list:We-acyclic3-mad}

    \item Find the co-actor groups with \textit{strong movie ratings} that have \textit{high interaction} with those director groups who also make movies with similar ratings (as that of co-actors.) 
    
    Cyclic 3-community: M $\Theta_{M,A}$ A $\Theta_{A,D}$ D $\Theta_{D,M}$ M;  $\omega_e$
    \label{list:We-cyclic3-madm}

\end{enumerate}

Based on the above analysis requirements, the layers for the IMDb data set are formed as follows. The first layer (Layer A) has actors as nodes (actor entity type with co-actor intra-layer relationship) which are connected  if  \textit{two actors have acted together in at least one movie.} The second layer (Layer D) has director nodes (director entity type with similar genre relationship for intra-layer edges) which are linked if they have \textit{directed movies with overlapping  genres (at least 50\%.)}  In the Movie layer (Layer M), the movie nodes are connected \textit{based on the range of average ratings in a specific range} (movie entity type with rating relationship.) The ratings are divided into 5 disjoint ranges - [0-2), [2-4), [4-6), [6-8) and [8-10]. The movies with missing rating information were not connected. 

The inter-layer edges have the following semantics. The directs-actor inter-layer edges ($L_{A,D}$) connect a director with an actor if he has directed that actor in a movie. The  directs-movie inter-layer relationship ($L_{D,M}$) connects a director to his movies in the movie-ratings layer. Finally, acts-in-a-movie relationship is captured by the inter-layer edges ($L_{A,M}$) which connects an actor with the movies the actor has acted in. This exercise also highlights the modeling using a HeMLN for the set of analysis requirements. We have purposely chosen a cyclic HeMLN for illustrative purposes. Figure~\ref{fig:sequence} shows these layers, two inter-layer connections along with some communities in each layer.

For a specific analysis, the characteristics of the communities connected in the bipartite graph need to be used as meta edge weight to get desired coupling. For example, \textit{maximum interaction} and \textit{most popular} in ~\ref{list:maximum-interaction}, \ref{list:We-acyclic3-adm}, \ref{list:We-acyclic3-mad} and \ref{list:We-cyclic3-madm}, is interpreted as the number of edges between the participating communities. In contrast, strongly connected groups interacting  with other strongly connected groups as in~\ref{list:density}, is interpreted to include community density as well. Versatility or most sought after is mapped to participation of hub nodes in each group as in  ~\ref{list:hub-participation}.

The purpose of a k-community is for analyzing arbitrary number of layers of a HeMLN. ~\ref{list:We-acyclic3-adm} requires a 3-community (for 3 layers) with an acyclic specification (using only 2 edges.) The 3-community analysis description of ~\ref{list:We-acyclic3-adm} also includes the order from the specification. The same is true for ~\ref{list:We-acyclic3-mad}, but the order starts with the Movie layer as the analysis is from that perspective. Finally, ~\ref{list:We-cyclic3-madm} requires a cyclic 3-community using inter-layer relationships between all layers in a particular order.
\section{Customizing Maximal Flow}
\label{sec:customize-maxflow}

Algorithm~\ref{alg:k-community} in Section~\ref{sec:community-detection} uses a bipartite graph match with a given weight metric. As we indicate earlier, there is an important difference between simple nodes and \textbf{meta nodes} that represent a \textbf{community of nodes and edges with their own characteristics}. Without bringing these characteristics of meta nodes to the maximal flow match, we cannot argue that the pairing obtained represents an analysis based on a participating community characteristics. Hence, it is important to identify how qualitative characteristics can be mapped quantitatively to a weight metric (or weight of the meta edge in community bipartite graph) to influence the bipartite matching. Below, we detail the three proposed weight metrics and intuition behind them.

\subsection{Number of Inter-Community Edges ($\omega_e$)}
\label{sec:m1}

This metric uses the number of inter-community edges of the participating communities as the weight (normalized) for the meta-edge connecting the two communities (meta nodes) in the community bipartite graph. The intuition behind this metric is \textit{maximum connectivity} (size of the community is to some extent factored into it) without including other characteristics of the communities. This weight connotes \textit{maximum interaction between two communities} acting as meta-nodes (used for ~\ref{list:maximum-interaction}, ~\ref{list:We-acyclic3-adm}-~\ref{list:We-cyclic3-madm}.)
\\

For every meta edge $(u_i^m, u_k^n)$ $\in$ $L'_{i,k}$, where $u_i^m$ and $u_k^n$ are the meta nodes corresponding to communities $c_i^m$ and $c_k^n$, respectively, in the community bipartite graph, the weight, 

\textit{\centerline{$\omega_e$($u_i^m$, $u_k^n$) =  $|x_{i,k}|$, ~where}} 
\textit{\noindent $x_{i,k}$  = $\bigcup~~ \{(a, b): a \in v_i^{c^m}, b \in v_k^{c^n}, ~and~ (a, b) \in L_{i,k}\}$. Each meta edge weight is normalized by dividing by max($\omega_e$).}

\subsection{Density and Edge Fraction ($\omega_d$)}
\label{sec:m2}

The intuition behind this metric is to not only bring the edge contribution as a fraction (instead of the total number of edges as in $\omega_e$.), but also participating community characteristics. Density which captures internal structure of a community is used.
Clearly, \textit{higher the densities and larger the edge fraction, the stronger is the interaction between two meta nodes (or communities.)} Since each of these three components (each being a fraction) increases the strength of the inter-layer coupling, they are  multiplied to generate the weight of the meta edge.  The domain of this weight will be $(0,1]$. Formally, using the density formula,
\\

For every $(u_i^m, u_k^n)$ $\in$ $L'_{i,j}$, where $u_i^m$ and $u_k^n$ denote the communities, $c_i^m$ and $c_k^n$ in the community bipartite graph, respectively, the weight,

\textit{\centerline{$\omega_d$($u_i^m$, $u_k^n$) =  $\frac{2*|e_i^{c^m}|}{|v_i^{c^m}|*(|v_i^{c^m}|-1)}$ * $\frac{|x_{i,k}|}{|v_i^{c^m}|*|v_k^{c^n}|}$*$\frac{2*|e_k^{c^n}|}{|v_k^{c^n}|*(|v_k^{c^n}|-1)}$,}
}
\textit{\noindent where $x_{i,k}$ = $\bigcup~~ \{(a, b): a \in v_i^{c^m}, b \in v_k^{c^n}, ~and~ (a, b) \in L_{i,k}\}$}

~\ref{list:density} uses this weight metric based on the analysis description.

\subsection{Hub Participation ($\omega_h$)}
\label{sec:m3}

The $\omega_e$ metric captures only interaction between two communities and the metric $\omega_d$ includes the effect of community structure, but not the characteristics of the nodes that interact. Typically, we are interested in knowing whether highly influential nodes within a community also interact across the community. This can be translated to the \textit{participation of influential nodes within and across each participating community} for analysis. This can be modeled by using the notion of  \textbf{hub}\footnote{High centrality nodes (or hubs) have been defined based on different metrics, such as degree centrality (vertex degree), closeness centrality (mean distance of the vertex from other vertices), betweenness centrality (fraction of shortest paths passing through the vertex), and eigenvector centrality.}  \textbf{participation} within a community and their interaction across layers. In this paper, we have used degree centrality for this metric to connote higher influence. Again, ratio of participating hubs from each community and the edge fraction are multiplied to compute $\omega_h$. Formally,
\\

For every $(u_i^m, u_k^n)$ $\in$ $L'_{i,k}$, where $u_i^m$ and $u_k^n$ denote the communities, $c_i^m$ and $c_k^n$ in the community bipartite graph, respectively, the weight,

\textit{\centerline{ $\omega_h$($u_i^m$, $u_k^n$) =  $\frac{|H_{i,k}^{m,n}|}{|H_i^m|}$ * $\frac{|x_{i,k}|}{|v_i^{c^m}|*|v_k^{c^n}|}$*$\frac{|H_{k,i}^{n,m}|}{|H_k^n|}$,}
}
\textit{\noindent where $x_{i,k}$ = $\bigcup~~\{(a, b): a \in v_i^{c^m}, b \in v_k^{c^n}, ~and~ (a, b) \in L_{i,j}\}$; $H_i^m$ and $H_k^n$ are set of hubs in $c_i^m$ and $c_k^n$, respectively; $H_{i,k}^{m,n}$ is the set of hubs from $c_i^m$ that are connected to $c_k^n$; $H_{k,i}^{n,m}$ is the set of hubs from $c_k^n$ that are connected to $c_i^m$ }.
\\ 

From the analysis description, ~\ref{list:hub-participation} has chosen this metric. 

In addition to the above, many other useful weight metrics have been explored to meet analysis requirements. Due to space constraints, we have included what we think are the most useful based on our  analysis experience.

\subsection{Uniqueness of Proposed Metrics}
\label{sec:compare-metrics}

Ideally, the alternatives for metrics should be independent of each other so they produce different analysis. Also, it is important that their computation be efficient. We believe that the three metrics proposed satisfy the above (see Section~\ref{sec:cost-analysis}.) The metrics make use of the \textit{pre-computed} community characteristics of the participating communities while deducing the weight of meta edge for bipartite match. However, as max-flow algorithm works by optimizing the global flow, independence of metrics does not imply unique maximal flow matches. This is corroborated by the experimental results as well. However, to show that the proposed metrics are unique, we can establish some boundary (or stringent) conditions under which some of them produce the same pairs as discussed below. Although the lemmas themselves are not that difficult to prove, they convey the uniqueness of the weight metrics.

\begin{description}
\item [Lemma~\ref{lemma:ed}] shows one of the rare conditions under which $\omega_e$ and $\omega_d$ produce the same results. \textit{Two layers each producing communities all of which are the same size cannot be common}!
\item[Lemma ~\ref{lemma:dh}] indicates that $\omega_d$ and $\omega_h$ produce the same results if the \textit{communities are not only cliques but all nodes in that community are also participating}. This again is a boundary condition.
\item [Lemma~\ref{lemma:edToh}] shows under what conditions, the three metrics produce the same results. Again, the \textit{conditions are extremely stringent} making these metrics unique,
\end{description}

\begin{lemma}
\label{lemma:ed}
\textit{For metrics $\omega_e$ and $\omega_d$, the 2-community results for layer $G_i$ and layer $G_j$ will be same if, \\
a). All communities of both layers are cliques, and \\
b). For each layer, all communities are equi-sized, that is they have same number of nodes}
\end{lemma}
\begin{proof}
The proof follows directly by equating the formula of the two metrics. In the scenario stated, for every meta edge in $CBG_{i,j}$, the weight with $\omega_d$ will be \textit{1/pq} times the weight with $\omega_e$, where \textit{p} and \textit{q} are the sizes of each community in $G_i$ and $G_j$, respectively. Thus, max-flow will output the same set of pairs, as all the weights change in a fixed proportion.    
\end{proof}

\begin{lemma}
\label{lemma:dh}
\textit{For metrics $\omega_d$ and $\omega_h$, the 2-community results for layer $G_i$ and $G_j$ will be same if, \\
a). All communities for both layers are cliques, and \\
b). For every community pair (say, $c_i^m$ and $c_j^n$), the product of the number of participating hubs is equal to the product of the number of nodes in the communities, i.e. $v_i^{c^m}$ * $v_j^{c^n}$ = $H_{i,j}^{m,n}$ * $H_{j,i}^{n,m}$}
\end{lemma}
\begin{proof}
The above conditions can be derived by equating the two formulae.
\end{proof}

\begin{lemma}
\label{lemma:edToh}
\textit{For layers $G_i$ and $G_j$, the common 2-community tuples with metrics $\omega_e$ and $\omega_d$, are generated for $\omega_h$, if \\
a). All communities for both layers are cliques and, \\
b). For every inter-layer edge between two communities, there is unique participation from both layers, \textbf{or}, the number of inter-layer edges where there is unique participation from both layers must be at least the minimum number of nodes from two layers}
\end{lemma}

\begin{proof}
Suppose, community $c_i^1$ from $G_i$ got matched to $c_j^1$ out of k possible options, $c_j^1$, $c_j^2$, ..., $c_j^k$ from $G_j$ using $\omega_e$ and $\omega_d$. This means, \\
(i) Out of all communities in $G_j$, $c_i^1$ has \textit{maximum interaction (inter-layer edges)} with $c_j^1$ (condition for $\omega_e$.) Moreover, every node is a hub (property of a clique - condition (a).)
Thus, because of the participation criteria highlighted in condition (b) \textit{$c_j^1$ is the one that has maximum participating hubs among the k alternatives}. \\
(ii) $c_i^1$ has \textit{highest edge fraction interaction} with $c_j^1$ from $G_j$ among the k alternatives (condition for $\omega_d$.)

Therefore, from (i) and (ii) it is deduced that the meta edge weight (in $CBG_{i,j}$ for metric $\omega_h$) between $c_i^1$ and $c_j^1$ will be the \textit{maximum among all the weights} assigned for the meta edges going from $c_i^1$. Thus, among the k community alternatives from $G_j$, the community $c_i^1$ will be matched to $c_j^1$ by the max-flow algorithm. For any other matching, the overall flow across the CBG will reduce and may even reduce the total number of matches.
\end{proof}

\subsection{k-community Detection Cost Components}
\label{sec:cost-analysis}

For a given specification of a k-community, its detection has several components each with its own cost. Below, we summarize their individual complexity and cost.

\begin{enumerate}
    \item Cost of generating 1-community for each layer (or a subset of layers needed) can be done in parallel, hence bounding this \textbf{one-time cost} to the largest one (typically for a layer with maximum density.)
    \item For the proposed analysis metrics, part of them are, again, \textbf{one-time costs} calculated independently on the results of 1-community. These  costs for $\omega_d$ and $\omega_h$ can be computed on a single pass of the communities  using their node/edge details generated by the community detection algorithm.

    \item The cost of a 2-community (base case) detection includes the cost of generating the bipartite graph, computing the weight of each meta edge of the community bipartite graph for a given $\omega$, and maximal network flow cost. This has to be done for each iteration as well. Only the edge fraction (or the maximum number of edges) and participating hubs need to be computed during the iteration. The cost of maximal network flow for the Edmonds-Karp implementation of Ford-Fulkerson method ~\cite{edmonds1972theoretical,ford_fulkerson_1956} used in our experiments is O($|V|*|E|^2$), where V and E are the number of nodes and edges in the community bipartite graph, respectively. 
    The bipartite graph can be generated during the computation of weights for the meta edges. Luckily, in our community bipartite graph, the number of meta edges is \textbf{order of magnitude less } than the number of edges between layers. Also, the number of meta nodes is bounded by the base case.
    
    \item The components costs within each iteration (and also the base case) are: i) weight computation, community bipartite graph creation, and maximal network flow cost. This cost is bounded by the base case (and will not increase) due to no or inconsistent matches as iterations progress. 

\end{enumerate}

In summary, for the proposed decoupling approach, the bulk of the cost is \textbf{one-time} (1-community detection and portions of $\omega_d ~and~ \omega_h$.) \textbf{Cost within each iteration is insignificant} compared to the one-time cost which is also borne out by the experimental analysis in Section~\ref{sec:experiments}.
\section{Experimental Analysis}
\label{sec:experiments}

In this section, we present the results of our  analysis of  ~\ref{list:maximum-interaction} through ~\ref{list:We-cyclic3-madm} discussed in Section~\ref{sec:application-and-analysis} on the IMDb data set~\cite{data/type2/IMDb} modeled as a HeMLN. To recap, we have 3 layers,  Actor (layer A), Director (layer D), and Movies (layer M) and 3 inter-layer relationships (directs-actor, acts-in-a-movie, and directs-movies.) We start with the individual layer analysis and  continue with all the 6 analysis using the approach proposed in this paper and discuss their corroboration with the intuitive expectations and analysis discussed earlier. 

\begin{figure}[h]
   \centering \includegraphics[width=\linewidth]{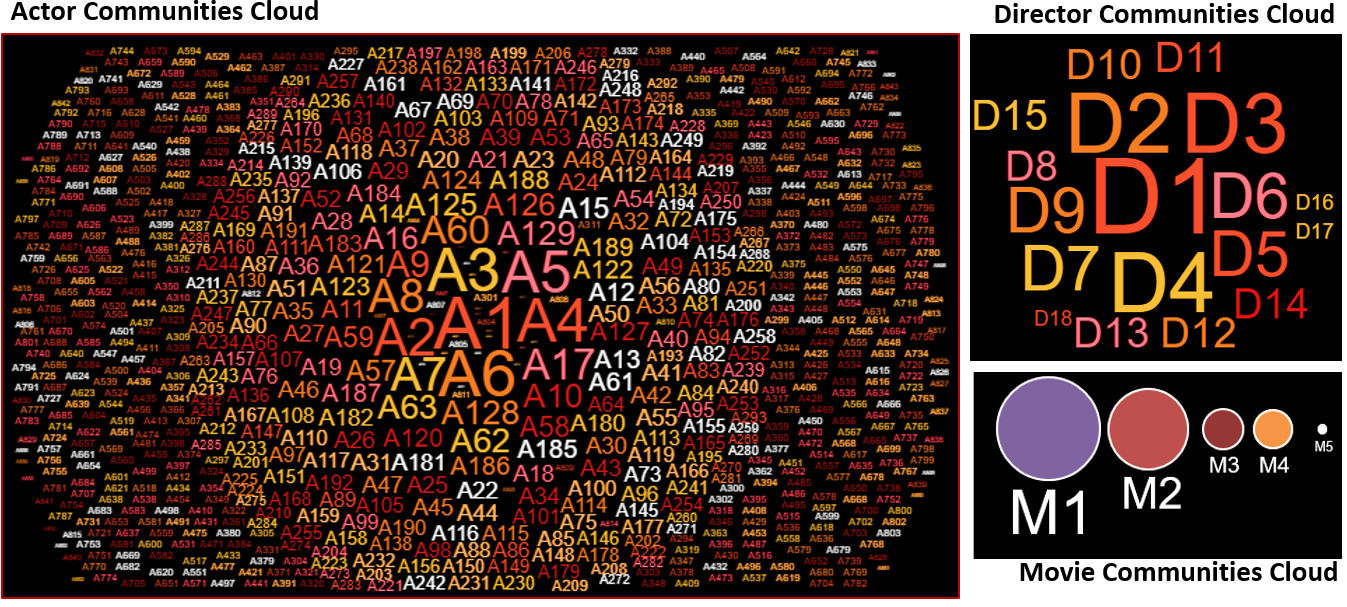}
    \caption{Community Clouds Based on Size}
   \label{fig:cloud}
\end{figure}

\textbf{Experiment Setup:} A random set of 1000 movies that released anywhere in the world in the period from 2001 to 2015 were chosen for building the HeMLN, comprising of 4588 actors and 1091 directors from that period. A 1-community is computed for each layer using the widely used package Infomap~\cite{InfoMap2014}, that works in a hierarchical fashion while optimizing the \textit{map equation}, which exploits the \textit{information - theoretic duality} between the problem of compressing data, and the problem of detecting and extracting significant patterns or structures (communities) within those data based on \textit{information flow}.  
We use the publicly available package from NetworkX for the bipartite graph match~\cite{Python:maximal-match}. The k-community detection Algorithm \ref{alg:k-community} was implemented in Python version 3.6 and was executed on a quad-core $8^{th}$ generation Intel i5 processor Windows 10 machine with 8 GB RAM. 
\subsection{IMD\MakeLowercase{b} Individual Layer Analysis}

For the purposes of this analysis, we generated disjoint 1-community for each individual layer.
The statistics for each layer are shown in Table~\ref{table:layer-stats}. The community size distribution is shown as a cloud for each layer in Figure~\ref{fig:cloud}.
For clarity, we show each layer with layer prefix and community id.

A number of insights can be gleaned from layer analysis. The average, mode and median actor community size are 5.3, 4 and 4, respectively. The proportion of communities of median size or less (529 out of 819 or at least 61\%) indicate that co-acting is limited to smaller groups of actors. There is only 1 large group (51) of co-actors. Further, 790 (92\%) clique communities indicate that actors who co-act do it with the same group.  In contrast,  the number of communities is proportionately less (with respect to the number of nodes) in the director layer indicating that directors do not limit themselves to directing a small set of genres. The lower number of communities in the Movie (rating) layer is expected as it is a layer of 5 non-overlapping rating ranges.  Also, 35\% of movies (major chunk) have ratings in the range [6-8), A small number of movies (community M5) has the lowest range of rating [0, 2). Also, we know that all of them are cliques.

\begin{table}[h!t]
\centering
    \begin{tabular}{|c|c|c|c|}
            \hline
                 & Actor & Director & Movie \\
            \hline
                \#Vertices & 4588 & 1091 & 1000 \\
            \hline
            Layer Density & 0.00097 & 0.15471 & 0.17516 \\
            \hline
            \#Comm. (Size $>$ 1) & 862 & 18 & 5 \\
            \hline
                \#Clique Communities & 790 & 7 & 5 \\
            \hline
                \% Vertices in Cliques & 82.6\% & 7.4\% &  100\%\\
            \hline
                Avg. Size & 5.3 & 60.61 & 132.5 \\                
            \hline
                Mode Size (\#Comm) & 4 (510) & 2 (3) & - \\
            \hline
                Median Size (\#Comm.) & 4 (510) & 31.5 (0) & 57 (0)\\
            \hline
                Max. Size (\#Comm.) & 51 (1) & 357 (1) & 352 (1)\\
            \hline
                Min. Size (\#Comm.) & 2 (19) & 2 (3) & 3 (1)\\
            \hline

        \end{tabular}
            \caption{1-community Statistics for IMDb HeMLN}
    \label{table:layer-stats}
\end{table}

\subsection{IMD\MakeLowercase{b} 2-community Analysis}
\label{sec:2-com}

With 3 layers and 3 weight metrics, there are a total of 9 possible 2-community specifications. As the movie-rating layer has very few (and only clique) communities, we have chosen the A and D layers (with A as the left set of the bipartite match) for analysis using different metrics as shown in ~\ref{list:maximum-interaction} to ~\ref{list:hub-participation}. As each metric is based on a different intuition  for inter-layer coupling, they are expected to give different results.

Table \ref{table:A1-results-allD}, \ref{table:A2-results-allD}, and \ref{table:A3-results-allD} show the results of 2-community as community pairings, respectively, for each metric.
When the entire 1-community  is used for each layer,  \textbf{44.4\% pairings (8 out of the 18) (marked in \textcolor{applegreen}{green})} are the \textit{same across metrics}. In all the common pairings, at least one of the participating community is a clique.

To understand the analysis differences across $\omega$, we grouped non-clique (density $<$ 1) communities from each layer for analysis.
Table \ref{table:A1-results-Dlt1}-\ref{table:A3-results-Dlt1}
shows the results for non-clique communities from each layer.
Just \textbf{one common pairing} is found (marked in \textcolor{cyan}{\textbf{cyan}}) \textbf{validating the uniqueness of proposed  metrics.}

For completeness, Table \ref{table:A1-results-Deq1}, \ref{table:A2-results-Deq1}, and \ref{table:A3-results-Deq1} show the community pairings when only cliques are used from each layer. Every matched pair that appears for $\omega_e$ (Table \ref{table:A1-results-Deq1}) and $\omega_d$ (Table \ref{table:A2-results-Deq1}) satisfied the conditions listed in Lemma \ref{lemma:edToh}, thus justifying their appearance for $\omega_h$ (Table \ref{table:A3-results-Deq1}) as well (marked in \textcolor{blue}{\textbf{blue}}.)

\begin{table}[H]
\renewcommand{\arraystretch}{1}
\vspace{-20pt}
\centering
\subfloat[]{

    \begin{tabular}{|c|c|}
        \hline
        \multicolumn{2}{|c|}{\textbf{\ref{list:maximum-interaction}}, $\omega_e$} \\
        \hline
        $c_A$ & $c_D$ \\
        \hline
        A6 & D1 \\
        \hline
        A252 & D2 \\
        \hline
        \textcolor{applegreen}{\textbf{A577}:c} & \textcolor{applegreen}{\textbf{D3}} \\
        \hline
        A2 & D4 \\
        \hline
        \textcolor{applegreen}{\textbf{A555}:c} & \textcolor{applegreen}{\textbf{D5}} \\
        \hline
        A100:c & D6 \\
        \hline
        \textcolor{applegreen}{\textbf{A374}:c} & \textcolor{applegreen}{\textbf{D7}:c} \\
        \hline
        \textcolor{applegreen}{\textbf{A484}:c} & \textcolor{applegreen}{\textbf{D8}} \\
        \hline
        A10 & D9 \\
        \hline
        A683:c & D10 \\
        \hline
        \textcolor{applegreen}{\textbf{A83}:c} & \textcolor{applegreen}{\textbf{D11}} \\
        \hline
        \textcolor{applegreen}{\textbf{A89}:c} & \textcolor{applegreen}{\textbf{D12}:c} \\
        \hline
        \textcolor{applegreen}{\textbf{A46}:c} & \textcolor{applegreen}{\textbf{D13}:c} \\
        \hline
        A161:c & D14:c \\
        \hline
        \textcolor{applegreen}{\textbf{A220}:c} & \textcolor{applegreen}{\textbf{D15}} \\
        \hline
        A85:c & D16:c \\
        \hline
        A188 & D17:c \\
        \hline
        A53:c & D18:c \\
        \hline
        \end{tabular}
        \label{table:A1-results-allD}
}
\subfloat[]{

    \begin{tabular}{|c|c|}
        \hline
        \multicolumn{2}{|c|}{\textbf{\ref{list:density}}, $\omega_d$} \\
        \hline
        $c_A$ & $c_D$ \\
        \hline
        A511:c & D1 \\
        \hline
        A204:c & D2 \\
        \hline
        \textcolor{applegreen}{\textbf{A577}:c} & \textcolor{applegreen}{\textbf{D3}} \\
        \hline
        A483:c & D4 \\
        \hline
        \textcolor{applegreen}{\textbf{A555}:c} & \textcolor{applegreen}{\textbf{D5}} \\
        \hline
        A332:c & D6 \\
        \hline
        \textcolor{applegreen}{\textbf{A374}:c} & \textcolor{applegreen}{\textbf{D7}:c} \\
        \hline
        \textcolor{applegreen}{\textbf{A484}:c} & \textcolor{applegreen}{\textbf{D8}} \\
        \hline
        A158:c & D9 \\
        \hline
        A683:c & D10 \\
        \hline
        \textcolor{applegreen}{\textbf{A83}:c} & \textcolor{applegreen}{\textbf{D11}} \\
        \hline
        \textcolor{applegreen}{\textbf{A89}:c} & \textcolor{applegreen}{\textbf{D12}:c} \\
        \hline
        \textcolor{applegreen}{\textbf{A46}:c} & \textcolor{applegreen}{\textbf{D13}:c} \\
        \hline
        A824:c & D14:c \\
        \hline
        \textcolor{applegreen}{\textbf{A220}:c} & \textcolor{applegreen}{\textbf{D15}} \\
        \hline
        A310:c & D16:c \\
        \hline
        A330:c & D17:c \\
        \hline
        A793:c & D18:c \\
        \hline
        \end{tabular}
        \label{table:A2-results-allD}
}
\subfloat[]{
    \begin{tabular}{|c|c|}
        \hline
         \multicolumn{2}{|c|}{\textbf{\ref{list:hub-participation}}, $\omega_h$} \\
        \hline
        $c_A$ & $c_D$ \\
        \hline
        A511:c & D1 \\
        \hline
        A204:c & D2 \\
        \hline
        \textcolor{applegreen}{\textbf{A577}:c} & \textcolor{applegreen}{\textbf{D3}} \\
        \hline
        A483:c & D4 \\
        \hline
        \textcolor{applegreen}{\textbf{A555}:c} & \textcolor{applegreen}{\textbf{D5}} \\
        \hline
        A332:c & D6 \\
        \hline
        \textcolor{applegreen}{\textbf{A374}:c} & \textcolor{applegreen}{\textbf{D7}:c} \\
        \hline
        \textcolor{applegreen}{\textbf{A484}:c} & \textcolor{applegreen}{\textbf{D8}} \\
        \hline
        A381:c & D9 \\
        \hline
        A828:c & D10 \\
        \hline
        \textcolor{applegreen}{\textbf{A83}:c} & \textcolor{applegreen}{\textbf{D11}} \\
        \hline
        \textcolor{applegreen}{\textbf{A89}:c} & \textcolor{applegreen}{\textbf{D12}:c} \\
        \hline
        \textcolor{applegreen}{\textbf{A46}:c} & \textcolor{applegreen}{\textbf{D13}:c} \\
        \hline
        A824:c & D14:c \\
        \hline
        \textcolor{applegreen}{\textbf{A220}:c} & \textcolor{applegreen}{\textbf{D15}} \\
        \hline
        A310:c & D16:c \\
        \hline
        A330:c & D17:c \\
        \hline
        A793:c & D18:c \\
        \hline        
        \end{tabular}

        \label{table:A3-results-allD}
}
\vspace{-12pt}
\caption*{\small{\textbf{All Communities}, c indicates a clique,  \textcolor{applegreen}{44.4\% common pairings} (862 A Communities, 18 D Communities)}}

\subfloat[]{

    \begin{tabular}{|c|c|}
        \hline
        \multicolumn{2}{|c|}{\textbf{\ref{list:maximum-interaction}}, $\omega_e$} \\
        \hline
        $c_A$ & $c_D$ \\
        \hline
        \textcolor{blue}{\textbf{A374}} & \textcolor{blue}{\textbf{D7}} \\
        \hline
        \textcolor{blue}{\textbf{A89}} & \textcolor{blue}{\textbf{D12}} \\
        \hline
        \textcolor{blue}{\textbf{A46}} & \textcolor{blue}{\textbf{D13}} \\
        \hline
        A161 & D14 \\
        \hline
        A85 & D16 \\
        \hline
        \textcolor{blue}{\textbf{A330}} & \textcolor{blue}{\textbf{D17}} \\
        \hline
        A53 & D18 \\
        \hline
        \end{tabular}
        \label{table:A1-results-Deq1}
}
\subfloat[]{

    \begin{tabular}{|c|c|}
        \hline
        \multicolumn{2}{|c|}{\textbf{\ref{list:density}}, $\omega_d$} \\
        \hline
        $c_A$ & $c_D$ \\
        \hline
        \textcolor{blue}{\textbf{A374}} & \textcolor{blue}{\textbf{D7}} \\
        \hline
        \textcolor{blue}{\textbf{A89}} & \textcolor{blue}{\textbf{D12}} \\
        \hline
        \textcolor{blue}{\textbf{A46}} & \textcolor{blue}{\textbf{D13}} \\
        \hline
        A824 & D14 \\
        \hline
        A310 & D16 \\
        \hline
        \textcolor{blue}{\textbf{A330}} & \textcolor{blue}{\textbf{D17}} \\
        \hline
        A793 & D18 \\
        \hline
        \end{tabular}
        \label{table:A2-results-Deq1}
}
\subfloat[]{
    \begin{tabular}{|c|c|}
        \hline
         \multicolumn{2}{|c|}{\textbf{\ref{list:hub-participation}}, $\omega_h$} \\
        \hline
        $c_A$ & $c_D$ \\
        \hline
        \textcolor{blue}{\textbf{A374}} & \textcolor{blue}{\textbf{D7}} \\
        \hline
        \textcolor{blue}{\textbf{A89}} & \textcolor{blue}{\textbf{D12}} \\
        \hline
        \textcolor{blue}{\textbf{A46}} & \textcolor{blue}{\textbf{D13}} \\
        \hline
        A824 & D14 \\
        \hline
        A310 & D16 \\
        \hline
        \textcolor{blue}{\textbf{A330}} & \textcolor{blue}{\textbf{D17}} \\
        \hline
        A793 & D18 \\
          \hline
        \end{tabular}

        \label{table:A3-results-Deq1}
}
\vspace{-12pt}
\caption*{\small{\textbf{Clique Communities only}, \textcolor{blue}{57.14\% common pairings} (790 A Communities, 7 D Communities)}}

\subfloat[]{

    \begin{tabular}{|c|c|}
        \hline
        \multicolumn{2}{|c|}{\textbf{\ref{list:maximum-interaction}}, $\omega_e$} \\
        \hline
        $c_A$ & $c_D$ \\
        \hline
        A6 & D1 \\
        \hline
        \textcolor{cyan}{\textbf{A252}} & \textcolor{cyan}{\textbf{D2}} \\
        \hline
        A1 & D3 \\
        \hline
        A2 & D4 \\
        \hline
        A187 & D6 \\
        \hline
        A4 & D8 \\
        \hline
        A10 & D9 \\
        \hline
        A17 & D11 \\
        \hline
        A129 & D15 \\
        \hline
        \end{tabular}
        \label{table:A1-results-Dlt1}
}
\subfloat[]{

    \begin{tabular}{|c|c|}
        \hline
        \multicolumn{2}{|c|}{\textbf{\ref{list:density}}, $\omega_d$} \\
        \hline
        $c_A$ & $c_D$ \\
        \hline
         A59 & D1 \\
        \hline
         \textcolor{cyan}{\textbf{A252}} & \textcolor{cyan}{\textbf{D2}} \\
        \hline
         A257 & D3 \\
        \hline
         A190 & D4 \\
        \hline
         A293 & D6 \\
        \hline
         A246 & D8 \\
        \hline
         A10 & D9 \\
        \hline
         A122 & D11 \\
        \hline
         A129 & D15 \\
        \hline
        \end{tabular}
        \label{table:A2-results-Dlt1}
}
\subfloat[]{
    \begin{tabular}{|c|c|}
        \hline
         \multicolumn{2}{|c|}{\textbf{\ref{list:hub-participation}}, $\omega_h$} \\
        \hline
        $c_A$ & $c_D$ \\
        \hline
        A59 & D1 \\
        \hline
         \textcolor{cyan}{\textbf{A252}} & \textcolor{cyan}{\textbf{D2}} \\
        \hline
        A184 & D3 \\
        \hline
        A190 & D4 \\
        \hline
        A293 & D6 \\
        \hline
        A4 & D8 \\
        \hline
        A16 & D9 \\
        \hline
        A122 & D11 \\
        \hline
        \end{tabular}

        \label{table:A3-results-Dlt1}
}
\vspace{-12pt}
\caption*{\small{\textbf{Non-clique Communities only}, \textcolor{cyan}{11.11\% common pairings} (72 A Communities, 11 D Communities)}}
\vspace{-8pt}
\caption{2-community results for (A $\Theta_{A,D}$ D)}
\label{table:2-community}
    
\end{table}

We know that the metric $\omega_e$ does not depend on the community characteristics, such as density and hub participation, unlike $\omega_d$ and $\omega_h$. This can also be validated from the results. All the actor communities that are part of matches for $\omega_d$ (Table \ref{table:A2-results-allD}) and $\omega_h$ (Table \ref{table:A3-results-allD}) are cliques, marked by \textit{`c'}, (which should be the case.) However, for the Table~\ref{table:A1-results-allD} which uses $\omega_e$, it is not the case validating our intuition. Even the removal of cliques does not effect the non-clique matches that were obtained for $\omega_e$ (Table \ref{table:A1-results-allD} and Table \ref{table:A1-results-Dlt1}.) 

\subsection{IMD\MakeLowercase{b} Acyclic 3-community Analysis}

For \ref{list:We-acyclic3-adm} and \ref{list:We-acyclic3-mad}, we need to generate \textit{acyclic} 3-community based on the specification which also includes order and metric $\omega_e$.
Specifically, for analysis \ref{list:We-acyclic3-adm} and ~\ref{list:We-acyclic3-mad},
the director and actor layer, become the left layer for community bipartite graph, respectively, during the second composition.

\begin{figure}[ht]
   \centering \includegraphics[width=0.8\linewidth]{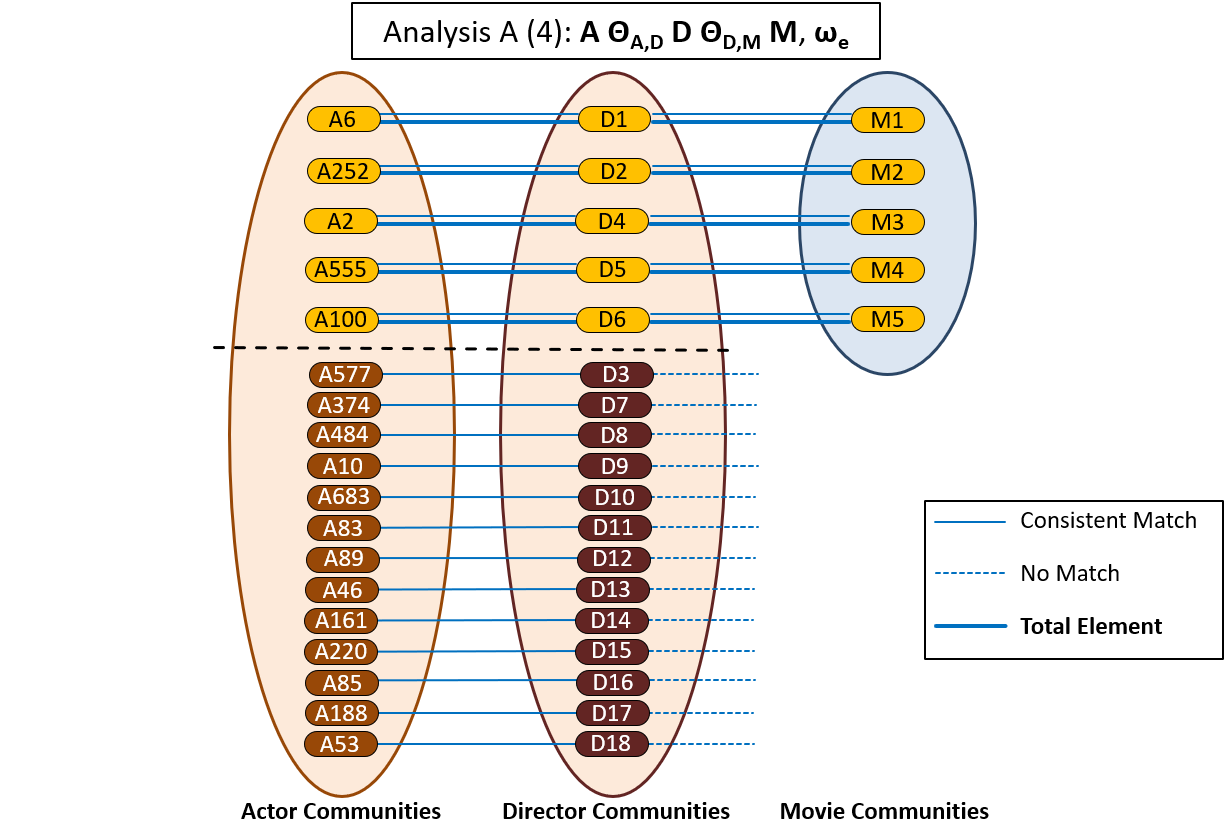}
   \caption{Acyclic 3-community (\textbf{total and partial})}
   \label{fig:A4-adm}
\end{figure}

Figures~\ref{fig:A4-adm} and ~\ref{fig:A5-mad}  show the 3-community analysis results diagrammatically for \ref{list:We-acyclic3-adm} and \ref{list:We-acyclic3-mad}, respectively. For analysis \ref{list:We-acyclic3-adm}, 18 elements (consistent matches) are obtained after the first composition, bounded by the 18 communities in the director layer. For the second composition, although all the 18 communities from layer director are carried over (as they paired during base case), only 5 
produce consistent matches with the movie layer, to get extended to become total elements (\textcolor{blue}{\textbf{bold blue line}}), whereas, for the other 13 there was a no match (\textcolor{blue}{broken blue lines}), thus becoming partial elements. 
Similarly, for \ref{list:We-acyclic3-mad}, every pairing from the first composition got extended in the second composition to produce 5 total elements (\textcolor{blue}{\textbf{bold blue lines}}), bounded by movie communities.  Moreover, for both the  result sets, none of the total elements overlap. Thus, \textbf{the \textit{order}} in which k layers are specified to obtain k-community determines the set of partial and total elements, providing insights corresponding to analysis specification.

\begin{figure}[ht]
   \centering \includegraphics[width=0.8\linewidth]{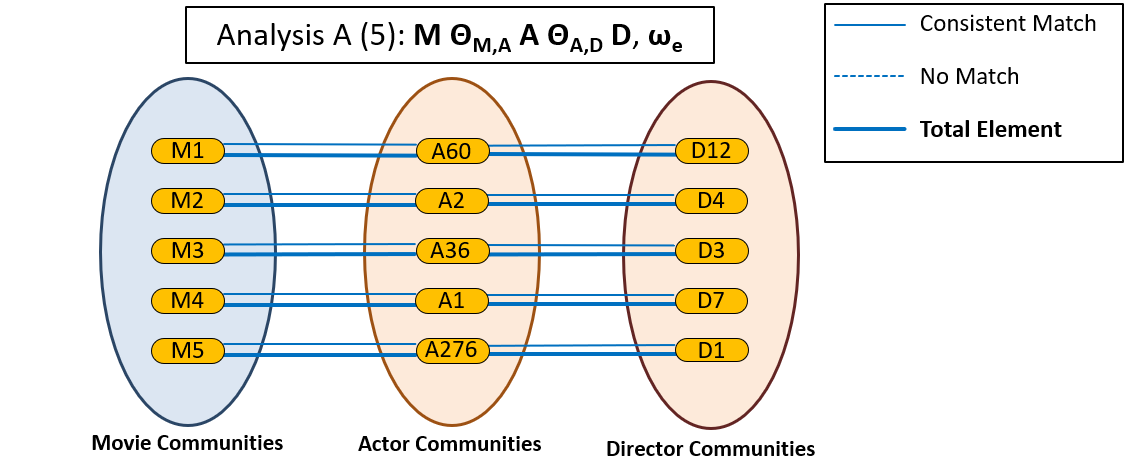}
   \caption{Acyclic 3-community (\textbf{no partial elements})}
   \label{fig:A5-mad}
\vspace{-10pt}
\end{figure}

\subsection{IMD\MakeLowercase{b} Cyclic 3-community Analysis}

Since 6 analysis orderings are possible for 3 layers and 3 sets of bipartite graphs forming a cycle (18 if metrics are included), choosing the order is important. \ref{list:We-cyclic3-madm} uses one of those orderings (M $\Theta_{M,A}$ A $\Theta_{A,D}$ D $\Theta_{D,M}$ M; $\omega_e$) where actors who are coupled with movie ratings collectively are further coupled with directors (giving directors who direct the popular actors for each rating.) This is coupled again with movies to check back whether the same director groups also have similar ratings. We have chosen $\omega_e$ for coupling.

\begin{figure}[h]
   \centering \includegraphics[width=0.8\linewidth]{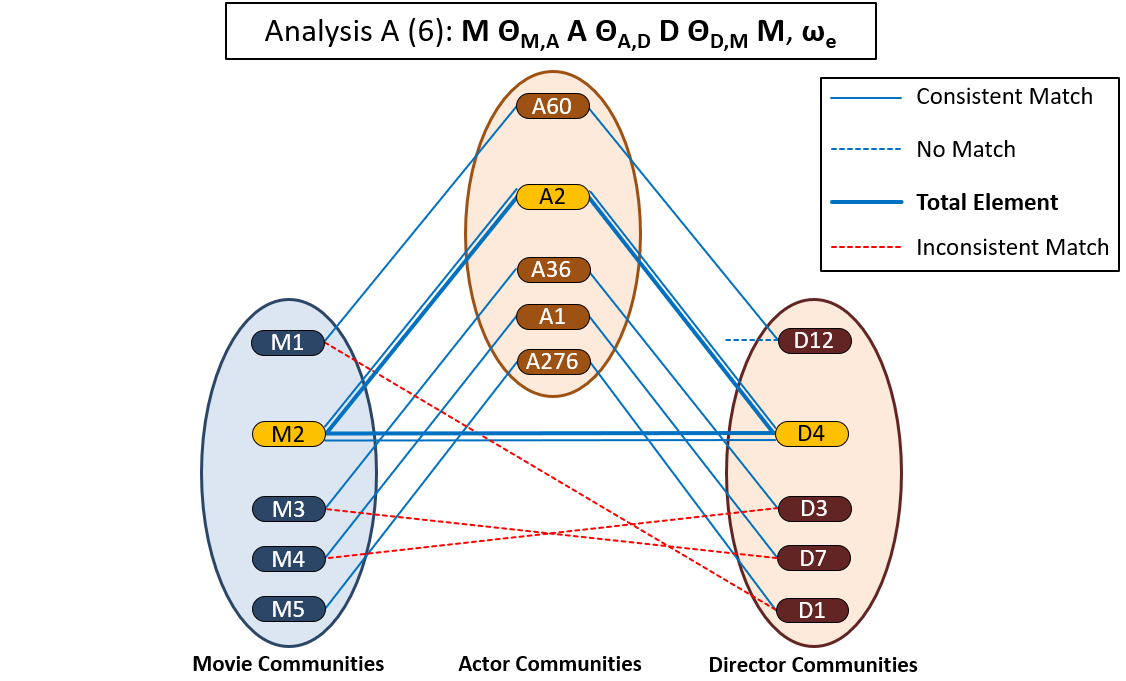}
   \caption{Progressive results of a cyclic 3-community}
   \label{fig:A6-madm}
\end{figure}

Results of each successive pairings (there are 3) are shown in Figure~\ref{fig:A6-madm} using the same color notation. Coupling of movie-actor pairs (first composition) results in 5 consistent matches bounded by the movie layer. \textbf{Also, it is easy to see from the figure that only one of them continues and becomes a total element for the cyclic 3-community (\textcolor{blue}{\textbf{bold blue triangle}}.) Further, the final result is an extension (M-A-D-M instead of M-A-D)} \textbf{of the acyclic 3-community result seen in Figure~\ref{fig:A5-mad}.}
When the base case is extended to the director layer (second composition), we got 5 consistent matches 
as can be seen in Figure~\ref{fig:A5-mad}.  The final composition to complete the cycle uses the 5 communities of director layer and 5 communities of the movie layer as left and right sets of community bipartite graph, respectively. \textbf{Only one consistent match is obtained to generate the total element (M2-A2-D4-M2) for the cyclic 3-community.} It is interesting to see 3 inconsistent matches (\textcolor{red}{red broken lines}) between the communities which clearly \textbf{indicate that all couplings are not satisfied by these pairs.} These result in 3 partial elements (M3-A36-D3, M4-A1-D7 and M5-A276-D1.)

\textbf{The inconsistent matches also highlight the importance of order which is fundamental to our k-community definition for analysis.} If a different order had been chosen (viz. director and actor layer as the base case), the result could have included the inconsistent matches. In this example, we also see \textit{one} no match (\textcolor{blue}{broken blue line}) in the final step, where D12 does not get matched to any movie community, thus generating the partial element, M1-A60-D12.

\subsection{Efficiency of Decoupled Approach}

\begin{figure}[ht]
   \centering \includegraphics[width=0.9\linewidth]{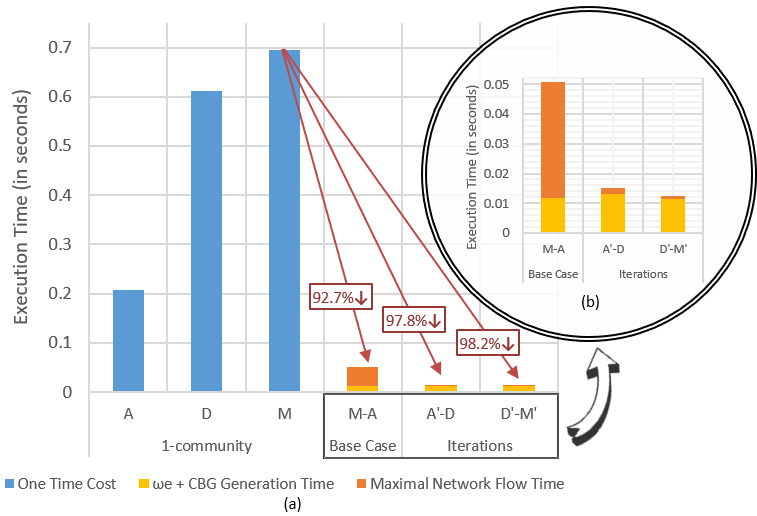}
   \caption{Performance Results for cyclic 3-community in \ref{list:We-cyclic3-madm}: (M $\Theta_{M,A}$ A $\Theta_{A,D}$ D $\Theta_{D,M}$ M; $\omega_e$)}
   \label{fig:perf-madm}
\end{figure}

The goal of the decoupling approach was to improve the efficiency of k-community detection using the divide and conquer approach. We illustrate that with the largest k-community we have computed which uses 3 iterations (including the base case.)
Figure \ref{fig:perf-madm}(a) shows the execution time for the one-time and iterative costs discussed earlier for \ref{list:We-cyclic3-madm}. The difference in one-time 1-community cost for the 3 layers follow their density shown in Table~\ref{table:layer-stats}. We can also see how the iterative cost is insignificant as compared to the one time cost (by an order of magnitude.) Iteration cost includes creating the bipartite graph, computing $\omega_e$ for meta edges, and maximal flow cost. As the iterations progress, the iterative cost decreases significantly as well. Even \textbf{the cost of all iterations together (0.078 sec) is still almost \textit{an order of magnitude less than the largest one-time cost} (0.694 sec for Movie layer.)} We have used this case as this subsumes all other cases. The zoomed in version (Figure \ref{fig:perf-madm} (b)) of the iterations further show how the iteration cost  reduces. As we had indicated, largest reduction comes for the iteration after the base case which can be clearly seen in Figure~\ref{fig:perf-madm} (b).

The \textbf{additional incremental cost for computing a k-community is extremely small validating the efficiency of decoupled approach}. 
\section{Conclusions}
\label{sec:conclusions}

In this paper, we have provided a structure-preserving definition of a k-community for a MLN, efficiency of its detection and versatility.  We adapted a composition function -  maximal network flow with customized weight metrics for a broader analysis. Finally, we used the approach for demonstrating its analysis capability and versatility using the IMDb data set.

\bibliographystyle{abbrv}
\bibliography{santraResearch,somu_research,itlabPublications}

\end{document}